\documentclass{llncs}
\usepackage[T1]{fontenc}
\usepackage[utf8]{inputenc}

\title{Simple Games versus Weighted Voting Games\thanks{This paper received support from the Leverhulme Trust (RPG-2016-258).}}

\author{Frits Hof\inst{1} \and Walter Kern\inst{1} \and Sascha Kurz\inst{2} \and Dani\"el~Paulusma\inst{3}}

\institute{University of Twente, The Netherlands, \texttt{f.hof@home.nl,w.kern@math.utwente.nl} 
\and  University of Bayreuth, Germany \texttt{sascha.kurz@uni-bayreuth.de}
\and Durham University, United Kingdom \texttt{daniel.paulusma@durham.ac.uk}}

\usepackage{amsmath,amssymb}
\usepackage{hyperref}
\usepackage{graphicx}
\usepackage{enumerate}
\usepackage{tikz}
\usepackage{boxedminipage}
\usetikzlibrary{arrows,shapes,calc}
\hypersetup{hypertexnames=false}

\def\R{{\mathbb R}}

\newcommand{\NP}{{\sf NP}}

\def\l{\lambda}

\begin{document}
\maketitle

\begin{abstract}
\noindent
A simple game $(N,v)$ is given by a set $N$ of $n$ players and a partition of~$2^N$ into  a set~$\mathcal{L}$ of losing coalitions~$L$ with value $v(L)=0$ that is closed under taking subsets and a set $\mathcal{W}$ of winning coalitions $W$ with $v(W)=1$.  Simple games with $\alpha= \min_{p\geq 0}\max_{W\in {\cal W},L\in {\cal L}} \frac{p(L)}{p(W)}<1$ are known as weighted voting games. Freixas and Kurz (IJGT, 2014) conjectured that $\alpha\leq \frac{1}{4}n$ for every simple game $(N,v)$. We confirm this conjecture for two complementary cases, namely when all minimal winning coalitions have size~$3$ and when no minimal winning coalition has size~$3$. As a general bound we prove that $\alpha\leq \frac{2}{7}n$ for every simple game $(N,v)$. For complete simple games, Freixas and Kurz conjectured that $\alpha=O(\sqrt{n})$. We prove this conjecture up to a $\ln n$ factor. We also prove that for graphic simple games, that is, simple games in which every minimal winning coalition has size~2, computing $\alpha$ is \NP-hard, but polynomial-time solvable if the underlying graph is bipartite. Moreover, we show that for every graphic simple game, deciding if $\alpha<a$ is polynomial-time solvable for every fixed $a>0$.
\end{abstract}

\section{Introduction}\label{sec_introduction}

Cooperative Game Theory provides a mathematical framework for capturing situations where subsets of agents  may form a coalition in order 
to obtain some collective profit or share some collective cost. Formally, 
a \emph{cooperative game  (with transferable utilities)} consists of a pair~$(N,v)$, where $N$ is a set of $n$ agents called \emph{players} and
$v: 2^N\to \R_+$ is a \emph{value function} that satisfies $v(\emptyset) = 0$.
In our context, the value $v(S)$ of a \emph{coalition} $S\subseteq N$ represents the profit for $S$ if all players in $S$ choose to collaborate with (only) each other. 
The central problem in cooperative game theory is to allocate the total profit $v(N)$ of the \emph{grand coalition}~$N$ to the 
individual players $i \in N$ in a {\lq\lq}fair{\rq\rq} way. To this end various \emph{solution concepts} such as the core, Shapley value or nuclueolus have been designed; see Chapter~9 of~\cite{peters2008games} for an overview. 

In our paper we study simple games~\cite{Sh62,NM44}.
 Simple games form a classical class of games, which are well studied;  see also the book of Taylor and Zwicker~\cite{taylor1999simple}.\footnote{Sometimes simple games are defined without requiring monotonicity (see, for example,~\cite{peters2008games}).}
The notion of being simple means that every coalition either has some equal amount of power or no power at all.  Formally, a cooperative game $(N,v)$ is \emph{simple} if  $v$ is a monotone 0--1 function with $v(\emptyset)=0$ and $v(N)=1$, so
$v(S)\in \{0,1\}$ for all $S\subseteq N$ and $v(S)\leq v(T)$ whenever $S\subseteq T$. 
 In other words, if $v$ is simple, then
there is a set $\mathcal{W} \subseteq 2^N$ of \emph{winning coalitions} $W$ that have value $v(W)=1$ and
a set $\mathcal{L} \subseteq 2^N$of \emph{losing coalitions} $L$ that have value $v(L)=0$. 
Note that $N\in {\cal W}$, $\emptyset \in {\cal L}$ and ${\cal W}\cup {\cal L}=2^N$. The monotonicity of $v$ implies that
subsets of losing coalitions
are losing and supersets of winning coalitions are winning. A winning coalition~$W$ is {\it minimal} if every proper subset 
of $W$ is losing, and a losing coalition~$L$ is {\it maximal} if every proper superset of $L$ is winning.

A simple game is a \emph{weighted voting game} if there exists a payoff vector $p \in \R_+^n$  
such that a coalition~$S$ is winning if $p(S)\geq 1$ and losing if $p(S)<1$. 
Weighted voting games are also known as {\it weighted majority games} and form one of the most popular classes of simple games.
Due to their practical applications in voting systems, computer operating systems and model resource allocation (see e.g.~\cite{BFJL02,CEW11}), structural and computational 
complexity aspects for solution concepts for weighted voting games have
been thoroughly investigated~\cite{ECJ08,EGGW09,freixas2011complexity,gvozdeva2013three}.
However, it is easy to construct simple games that are not weighted voting games. We give an example below, but in fact there are many important simple games that are not weighted voting games, and the relationship between weighted voting games and simple games is 
not yet fully understood. Therefore, Gvozdeva, Hemaspaandra, and Slinko~\cite{gvozdeva2013three} introduced a parameter $\alpha$, called the 
{\it critical threshold value}, to measure the {\lq\lq}distance{\rq\rq} of a 
simple game to the class of weighted voting games:
\begin{equation} \label{eq_minmax}
\alpha ~= ~\alpha(N,v)~=\min_{p \ge 0}~ \max_{W,L}~ \frac{p(L)}{p(W)}, 
\end{equation}
\noindent
where the maximum is taken over all winning coalitions in ${\cal W}$ and all losing coalitions in~${\cal L}$. 
A simple game $(N,v)$ is a weighted voting game if and only if $\alpha<1$.\footnote{If $\alpha\leq 1$, we speak of roughly weighted voting games~\cite{taylor1999simple}.} 
This follows from observing that each optimal solution $p$ of 
(\ref{eq_minmax}) can be scaled to satisfy $p(W)\ge 1$ for all winning coalitions $W$.

A concrete example of a simple game $(N,v)$ that is not a weighted voting game and that has in fact a large value of $\alpha$ was given 
in~\cite{paper_alpha_roughly_weighted}.
Let $N=\{1, \dots, n\}$ for some even integer~$n\geq 4$, and let the minimal winning coalitions be the pairs $\{1,2\}, \{2,3\}, \dots \{n-1, n\}, \{n,1\}$.
Consider any  payoff $p\ge 0$ satisfying $p(W) \ge 1$ for every winning coalition~$W$. Then $p_i+p_{i+1}\geq 1$ for $i=1,\ldots,n$ 
(where $n+1=1$). 
This means that $p(N)\geq \frac{1}{2}n$. 
Then, for at least one of $L=\{2,4,6,\dots, n\}$ and $L=\{1,3,5, \dots, n-1\}$, we have  
$p(L)\geq \frac{1}{4}n$, showing that $\alpha \ge \frac{1}{4}n$. On the other hand, it is easily seen that $p \equiv \frac{1}{2}$ satisfies $p(W) \ge 1$
for all  winning coalitions and $p(L) \le \frac{1}{4}n$ for all losing coalitions, showing that $\alpha \le \frac{1}{4}n$. Thus 
we conclude that
 $\alpha=\frac{1}{4}n$.
Due to this somewhat extreme example, the authors of~\cite{paper_alpha_roughly_weighted} conjectured that $\alpha\leq \frac{1}{4}n$ for all simple games. 
This conjecture turns out to be an  interesting combinatorial problem.

\smallskip
\noindent
{\it Conjecture~1~\cite{paper_alpha_roughly_weighted}.} For every simple game $(N,v)$, it holds that $\alpha \le \frac{1}{4}n$.

\subsection{Our Results}

In Section~\ref{sec_3ornot3} we prove that Conjecture~1 holds for the case where all
minimal winning coalitions have size~$3$ and for its complementary case where no minimal winning collection has size~$3$.
We were not able to prove Conjecture~1 for all simple games.
However, in Section~\ref{sec_general} we show that $\alpha \le \frac{2}{7}n \approx 0.2858 n$ for every simple game. 

In  Section~\ref{sec_complete} we consider a subclass of simple games based on a natural desirability order~\cite{Is56}. A simple game $(N,v)$ is \emph{complete} if the players can be ordered by a complete, transitive ordering $\succeq$, say,
$1 \succeq 2 \succeq \dots \succeq n$, indicating that higher ranked players have more power (and are more desirable) than lower ranked players.
More precisely, $i \succeq j$ means that $v(S\cup \{i\}) \ge v(S \cup \{j\})$ for any coalition $S \subseteq N\backslash \{i,j\}$.
The class of complete simple games properly contains all weighted voting games~\cite{FP08}.
For complete simple games, we show a lower bound on $\alpha$
that is asymptotically lower than $\frac{1}{4}n$, 
namely $\alpha=O(\sqrt{n}\ln n)$.
This bound matches,  up to a $\ln n$ factor, the lower bound of $\Omega(\sqrt{n})$ in~\cite{paper_alpha_roughly_weighted} 
(conjectured to be tight in~\cite{paper_alpha_roughly_weighted}).

In Section \ref{sec_algo} we discuss some algorithmic and complexity issues. We focus on instances where all minimal winning coalitions have size~$2$.  We say that such simple games are {\it graphic}, as they can conveniently be described by a graph $G=(N,E)$
with vertex set $N$ and edge set $E~=~\{ij~|~\{i,j\}~ \text{is winning} \}$.
For graphic simple games we show that computing~$\alpha$ 
is \NP-hard in general (see below for some related results).
 On the positive side, we show that computing $\alpha$ is 
 polynomial-time solvable if the underlying graph $G=(N,E)$ is bipartite, or if $\alpha$ is known to be small  (less than a fixed number $a$). 
We conclude with some remarks and open problems in  Section~\ref{sec_conclusion}.

\subsection{Related Work}

Another way to measure the distance of a simple game to the class of weighted voting games is to use the {\it dimension} of a simple game~\cite{TZ93}, which is the smallest number 
of weighted voting games whose intersection equals a given simple game. 
However, computing the dimension of a simple game is \NP-hard~\cite{DW06}, and
the largest dimension of a simple game with $n$ players is $2^{n-o(n)}$~\cite{high_dimension}. 
Moreover,
 simple games with dimension~1 have $\alpha=1$, but
 $\alpha$ may be arbitrarily large for simple games with dimension larger than~1.\footnote{A simple game with $\frac{1}{2}n$ players of type A and $\frac{1}{2}n$ players of type B and minimal winning coalitions consisting of one player of each type has dimension~2 and $\alpha=\frac{1}{4}n$.}
 Hence there is no direct relation between the two distance measures.
We also note that 
Gvozdeva, Hemaspaandra, and Slinko~\cite{gvozdeva2013three} introduced two other distance parameters. One measures the power balance between small and large coalitions. The other one allows multiple thresholds instead of threshold~1 only. 
See \cite{gvozdeva2013three} for further details.

For graphic simple games, it is natural to take the number of players $n$ as the input size for answering complexity questions, but in general simple games may have different representations. For instance, one can list all minimal winning coalitions or all maximal losing coalitions. Under these two representations the problem of deciding if 
 $\alpha < 1$, that is, if a given simple game is a weighted voting game, is also polynomial-time solvable.
This follows from results of Heged\"us and Megiddo~\cite{hegedus1996geometric} and  Peled and Simeone~\cite{peled1985polynomial}, as shown by Freixas, Molinero, Olsen and Serna~\cite{freixas2011complexity}.
The latter authors also showed that the same result holds if the representation is given by listing all winning coalitions or all losing coalitions. Moreover, they gave a number of complexity results of recognizing other subclasses of simple games.

We also note a similarity of our research with research into matching games. 
In Section~\ref{sec_3ornot3} we show that a crucial case in our 
study is when the simple game is graphic, that is, defined on some graph $G=(N,E)$.
In the corresponding \emph{matching game} a coalition $S 
\subseteq N$ has value $v(S)$ equal to the maximum size of a matching in the subgraph of $G$ induced by $S$. 
One of the most prominent solution concepts is the \emph{core} of a game, defined by $core(N,v) := \{p \in \R^n ~|~ p(N)=v(N), ~p(S) \ge v(S)~ \forall S  \subseteq N\}$. 
A core allocation is stable, as no coalition has any incentive to object against it. However, the core may be empty.
Matching games are not simple games. Yet their core constraints are readily seen to simplify to $p\ge 0$ and $p_i+p_j \ge 1$ for all $ij \in E$.
Classical solution concepts, such as the core and core-related ones like least core,  nucleolus or nucleon are
well studied for matching games, see, for example,~\cite{biro2012matching,bock2015stable,faigle1998,kern2003matching,KKT,SR94}. However, the problems encountered there differ with respect to the objective function. 
For graphic simple games  we aim to bound $p(L)$  over all losing coalitions, subject to  $p \ge 0$,~$p_i+p_j \ge 1$ for all $ij \in E$, whereas for matching games with an empty core we wish to bound $p(N)$, subject to $p \ge 0$,~$p_i+p_j \ge 1$ for all $ij \in E$. 
Nevertheless, basic tools from matching theory like the Gallai-Edmonds decomposition play a role in both cases.

\section{Two Complementary Cases}\label{sec_3ornot3}

In this section we will consider the following two {\lq\lq}complementary{\rq\rq} cases: when all winning coalitions have size equal to~$3$ (Section~\ref{s-1}), and when no
winning coalition has size equal to~$3$ (Section~\ref{s-2}). First observe that winning coalitions of size~$1$ do not cause any problems. If $\{i\}$
is a winning coalition of size~$1$, we satisfy it by setting $p_i=1$. Since no losing coalition $L$ contains $i$, we may 
remove~$i$ from the game and solve (\ref{eq_minmax}) with respect to the resulting subgame. A similar argument applies if some 
$i \in N$ is not contained in any minimal winning coalition. We then simply define $p_i=0$ and remove $i$ from the game. 
Thus, we may assume without loss of generality that all minimal winning coalitions have size at least~$2$ and that they cover all of $N$.

\subsection{All Minimal Winning Coalitions Have Size $2$.}\label{s-1}
We first investigate the case where all minimal winning coalitions have size exactly~$2$.
This case (which is a crucial case in our study) can conveniently be translated to a graph-theoretic problem. Let $G=(N,E)$ be
the graph with vertex set $N$ whose edges are exactly the minimal winning coalitions of size~$2$ in our game $(N,v)$. Our assumption that 
$N$ is completely covered by minimal winning coalitions  means that $G$ has no isolated vertices. 
Losing coalitions correspond to independent sets of vertices $L\subseteq N$. Then the min max problem (\ref{eq_minmax}) becomes

\begin{equation}
 \alpha~:=~\alpha_G ~:=~ \min_p ~ \max_L ~ p(L),
\end{equation}
\noindent
where the minimum is taken over all {\it feasible} pay-off vectors $p$, that is, $p\in \R_+^n$ with $p_i+p_j \ge 1$ for every $ij\in E$, and the maximum is taken over all independent sets $L\subseteq N$.

We first consider the case where $G=(A \cup B, E)$ is bipartite. To explain the basic idea, we introduce the following concept
(illustrated in Figure \ref{fig:wellspread}).

\begin{figure}
\center
\begin{tikzpicture}[scale=.5]
\node[scale=1] (A) at (-5,0){A};
\node[scale=1] (S) at (0,0){S};
\node[scale=1] (B) at (-7,-4){B};
\node[scale=0.9] ({N(S)}) at (0,-4){N(S)};
  \draw[line width=0.5mm] (0,0) ellipse (8cm and 0.7cm);
  \draw[line width=0.5mm] (0,0) ellipse (3cm and 0.35cm);
  \draw[line width=0.5mm] (0,-4) ellipse (10cm and 0.8cm);
  \draw[line width=0.5mm] (0,-4) ellipse (5cm and 0.42cm);
\draw[line width=1.2] (-3,0)--(-5,-4); 
\draw[line width=1.2] (-8,0)--(-10,-4);
\draw[line width=1.2] (3,0)--(5,-4); 
\draw[line width=1.2] (8,0)--(10,-4);
\end{tikzpicture}
\caption{A well-spread bipartite graph.} \label{fig:wellspread}
\end{figure}
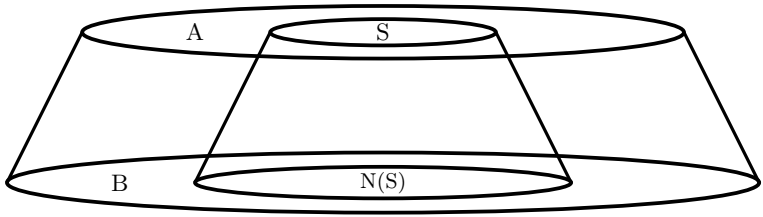

\noindent
{\bf Definition.}
Let $G=(A \cup B,E)$ be a bipartite graph of order $n = |A|+|B|$ without isolated notes and assume without loss of generality that $|A| \le |B|$. 
Let $\lambda \le \frac{1}{2}$ such that $|A|=\lambda n$ (and $|B|=(1-\lambda)n$).
 We say that $G$ is \emph{well-spread} with parameter $\lambda$ if for all
$S\subseteq A$ we have $$\frac{|S|}{|N(S)|} ~\le ~\frac{|A|}{|B|} ~=~ \frac{\lambda}{1-\lambda}.$$
(Here, as usual, $N(S) \subseteq B$ denotes the set of neighbors of $S$ in $B$.)

\smallskip
\noindent
Examples of well-spread bipartite graphs are biregular graphs or biregular graphs minus an edge. Note that if $G$ is 
well-spread with parameter $\lambda \le \frac{1}{2}$, then Hall's condition $|N(S)| \ge |S|$ for 
all  $S \subseteq A$ is satisfied, implying that $A$ can be completely matched to $B$ (see, for example,~\cite{lovasz2009matching}).
The following lemma is the key observation.

\begin{lemma}\label{lemma_well}
 Let $G=(A\cup B,E)$ be well-spread with parameter $\lambda\le \frac{1}{2}$. Then $p\equiv \lambda$ on $B$ and $p \equiv 1-\lambda$  on $A$ yields
 $\alpha_G \le \frac{1}{4}n$.
 \end{lemma}
 
\begin{proof}
Assume $L\subseteq N$ is an independent set. Let $\rho \le 1$ such that $|L \cap A| =\rho\lambda n$. Since $G$ is well-spread, we get
$|N(L\cap A)| \ge \rho (1- \lambda)n$, so that $|L \cap B| \le (1-\rho)(1-\lambda)n$. Thus
\[\begin{array}{lcl}
 p(L) &= &|L\cap A|(1-\lambda)+|L\cap B|\lambda\\[3pt]
         &\le &\rho\lambda n (1-\lambda) + (1-\rho)(1-\lambda) n \lambda\\[3pt]
         & \le &\rho \frac{1}{4}n+(1-\rho)\frac{1}{4}n\\[3pt]
         &\le &\frac{1}{4}n.
         \end{array}\]
Hence we have proven the lemma.   \qed
\end{proof}

In general, when $G=(A \cup B,E)$ is not well-spread, we seek to decompose $G$ into well-spread induced subgraphs $G_i=(A_i \cup B_i, E_i)$
with $A= \bigcup A_i$ and $B=\bigcup B_i$. Of course, this can only work if $G=(A \cup B,E)$ is such that $A$ can be matched
to $B$ in~$G$. 

\begin{proposition}\label{prop_bipartite}
Let $G=(A \cup B,E)$ be a bipartite graph without isolated vertices and assume that $A$ can be matched into $B$. Then $G$ decomposes
into well-spread induced subgraphs $G_i=(A_i \cup B_i, E_i)$, with $A= \bigcup A_i$ and $B=\bigcup B_i$
in such a way that for all $i,j$ with $i<j$,  $\lambda_i \ge \lambda_j$ and no edges join $A_i$ to $B_j$. \qed
\end{proposition}

\begin{proof}
Let $S \subseteq A$ maximize $|S|/|N(S)$. Set $A_1:=S$ and $B_1:=N(S)$. 
Let $G'$ be the subgraph of $G$ induced by $A\backslash A_1$
and $B':= B \backslash B_1$. Then $G'$ satisfies the assumption of the Proposition. Indeed, if $A'$ cannot be matched into $B'$ in $G'$, then
there must be some $S'\subseteq A'$ with $|S'| > |N'(S')|$, where $N'(S')=N(S')\backslash B_1$ is the neighborhood of $S'$ in $G'$. But then 
$|S \cup S'| =|S|+|S'|$ and $|N(S\cup S')|\le |N(S)|+|N'(S)|$ shows that $S$ cannot maximize $|S|/|N(S)|$, a contradiction. Thus, by 
induction, we may assume that $G'$ decomposes in the desired way into well-spread subgraphs $G_2, \dots, G_k$ with parameters
$\l_2 \ge \dots \ge \l_k$. The claim then follows by observing that 
$(i)$ no edges join $B_1$ to $A'$; and
$(ii)$ $\l_1 \ge \l_2$ (otherwise $ S \cup A_2$ would contradict the choice of $S$ maximizing $|S|/|N(S)|$). \qed
\end{proof}

We now combining the last two results.

\begin{corollary}\label{cor_bipartite}
 For every bipartite graph $G=(A\cup B,E)$ of order $n$ satisfying the assumption of Proposition \ref{prop_bipartite}, there exists a payoff 
 vector $p\ge 0$ 
 such that $p_i+p_j \ge 1$ for $ij\in E$ and $p(L)\le \frac{1}{4}n$ for any independent set $L \subseteq A \cup B$. In addition,
 $p$ can be chosen so as to satisfy $p \ge \frac{1}{2}$ on $A$.
\end{corollary}

\begin{proof}
The result follow immediately from Lemma~\ref{lemma_well} and Proposition~\ref{prop_bipartite}.
Note that if $p$ is chosen as $p\equiv 1-\l_i$ on $A_i$, then $p \ge \frac{1}{2}$ indeed. \qed
\end{proof}

As we will see, the assumption of Proposition \ref{prop_bipartite} is not really restrictive for our purposes. 
A (connected) component $C$ of a graph~$G$ is 
 {\it even} ({\it odd}) if  $C$ has an even (odd) number of vertices.
A graph~$G=(N,E)$ is \emph{factor-critical} if for every vertex $v\in V(G)$, the graph $G-v$ has a perfect matching.
We  recall the well-known Gallai--Edmonds Theorem 
(see~\cite{lovasz2009matching})
for characterizing the structure of maximum matchings in $G$; see also Figure~\ref{fig:decomp2}.
There exists a (unique) subset $A\subseteq N$, called a \emph{Tutte set}, such that
\begin{itemize}
 \item every even component of $G-A$ has a perfect matching;
 \item every odd component of $G-A$ is factor-critical; 
 \item every maximum matching in $G$ is the union of a perfect matching in each even component, a nearly perfect matching in each odd component and a matching that matches $A$ (completely) to the odd components.
\end{itemize}

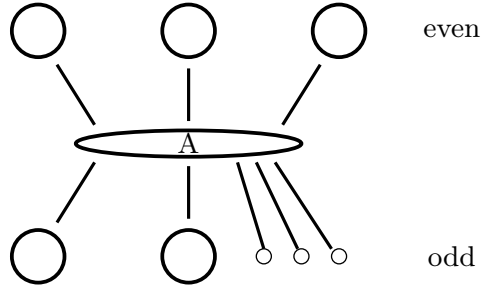
\begin{figure}
\center
\begin{tikzpicture}[scale=.5]

\node[scale=1.2] (A) at (0,0){A};
\node[scale=1.2] (E) at (7,3){\text{even}};
\node[scale=1.2] (O) at (7,-3){\text{odd}};
\draw[line width=.5]  (2,-3) circle [radius=0.2];
\draw[line width=.5]  (3,-3) circle [radius=0.2];
\draw[line width=.5]  (4,-3) circle [radius=0.2];
  \draw[line width=0.5mm] (0,0) ellipse (3cm and 0.35cm);
  
\draw[line width=0.5mm] (-4,3) circle (7mm);
\draw[line width=0.5mm] (0,3) circle (7mm);
\draw[line width=0.5mm] (4,3) circle (7mm);
\draw[line width=0.5mm] (-4,-3) circle (7mm);
\draw[line width=0.5mm] (0,-3) circle (7mm);

\draw[line width=1.2] (-2.5,0.5)--(-3.5,2.1);
\draw[line width=1.2] (2.5,0.5)--(3.5,2.1);
\draw[line width=1.2] (0,0.6)--(0,2); 
\draw[line width=1.2] (0,-0.6)--(0,-2); 
\draw[line width=1.2] (1.3,-0.5)--(2,-2.8); 
\draw[line width=1.2] (1.8,-0.5)--(2.9,-2.8); 
\draw[line width=1.2] (2.3,-0.5)--(3.8,-2.8); 
\draw[line width=1.2] (-2.5,-0.5)--(-3.5,-2.1);
\end{tikzpicture}
\caption{Tutte set $A$ splitting $G$ into even and odd components (possibly single nodes).} \label{fig:decomp2}
\end{figure}

\noindent
We are now ready to derive our first main 
result.\footnote{For $n$ is odd, the upper bound in Theorem \ref{thm_graph} can be slightly strengthened to
$\frac{n^2-1}{4n}$~\cite{hof2016weight}.}

\begin{theorem}\label{thm_graph}
 Let $G=(N,E)$ be a graph of order $n$. Then $\alpha_G \le \frac{1}{4}n.$ 
\end{theorem}

\begin{proof}
 Let $A\subseteq N$ be a Tutte set. Contract each odd component in $G-A$ to a single vertex and let $B$ denote the resulting
 set of vertices. The subgraph $\bar{G}$ induced by $A \cup B$ then satisfies the assumption of Corollary \ref{cor_bipartite}. Let 
 $\bar{p} \in \R^{|A|+|B|}$ be the corresponding payoff vector.  We define $p \in \mathbb{R}^n$ by setting
 $p_i=\bar{p}_i$ for every vertex $i \in A$ and every vertex~$i$ that corresponds to an odd component of size~$1$ in $G-A$.
 All other vertices get $p_j=\frac{1}{2}$.
 
 It is straightforward to check that $p \ge 0$ and $p_i+p_j \ge 1$. Indeed, $\bar{p} \ge \frac{1}{2}$ everywhere except on $B$, so
 the only critical edges $ij$ have $i \in A$ and $j$ a singleton odd component. But in this case $p_i+p_j =\bar{p}_i+\bar{p}_j \ge 1$.
 Thus we are left to prove that for every independent set $L \subseteq N$, $p(L) \le \frac{1}{4}n$. Let $B_0$ denote the set
 of singleton odd components $i \in B$, $L_0 := (L \cap A) \cup (L\cap B_0)$ and $n_0 := |A|+|B|$. Clearly, $L_0$ is an independent 
 set in the bipartite graph  $\bar{G}$ , and $p=\bar{p}$ on $L_0$. We thus conclude that $p(L_0) \le \frac{1}{4}n_0$.
 
 Next let us analyze $L\cap C$ where $C \subseteq N\backslash A$ is an even component.  $C$ is perfectly
 matchable, implying that $L$ contains at most $|C|/2$ vertices of $C$. So $p(L\cap C) \le\frac{1}{4}|C|$.
 A similar argument applies to odd components. Let $C$ be an odd component in $G-A$ of size at least~$3$. Then certainly
 $L$ cannot contain all vertices of $C$, so there exists some $i \in C\backslash L$. Since $C$ is factor-critical, $C\backslash i$ is 
 perfectly matchable, 
 implying that $L$ can contain at most half of $C\backslash i$. Thus $|L \cap C| \le (|C|-1)/2$ and
 $p(L\cap C) \le (|C|-1)/4$. 
 
 Summarizing, $n-n_0 = |N| - (|A|+|B|)$ is the sum over all $|C|$, where $C$ is an even component plus the sum over all
 $|C|-1$ where $C$ is an odd component, and $p(L \backslash L_0)$ is at most a $\frac{1}{4}$ fraction of this,
 finishing the proof. \qed
\end{proof}
 
We like to mention that both decompositions that we use to define the payoff $p$ can be computed efficiently. For the Edmonds--Gallai 
decomposition, this is a well-known fact (see, for example, \cite{lovasz2009matching}). For the decomposition into well-spread
subgraphs, this follows from the observation that deciding whether $\max_S \frac{|S|}{|N(S)|} \le r$ is equivalent to 
$\min_S r|N(S)|-|S| \ge 0$, which amounts to minimizing the submodular function $f(S)=r|N(S)|-|S|$; 
see, for example,~\cite{Sc00} for a strongly polynomial-time algorithm or Appendix~\ref{a-proof}. 

\subsection{No Minimal Winning Sets of Size~3}\label{s-2}
We now deal shortly with the more general case where there are, in addition, minimal winning coalitions of size~$4$ or larger. First 
recall how the
payoff $p$ that we proposed in Corollary \ref{cor_bipartite} works. For a bipartite graph $G=(A\cup B,E)$, split into well-spread
subgraphs $G_i=(A_i \cup B_i, E_i)$ with parameter $\l_i$, we let $p \equiv \l_i$ on $B_i$. So for $\l_i < \frac{1}{4}$, $p$ may be 
infeasible, that is, we may encounter winning coalitions $W$ of size~$4$ or larger with $p(W) <1$. This problem can easily be
remedied by raising $p$ a bit on each $B_i$ and decreasing it accordingly on $A_i$. Indeed, the standard $(\l, 1-\l)$ allocation
rule proposed in Lemma \ref{lemma_well} is based on the simple fact that $\l(1-\l) \le \frac{1}{4}$, which gives us some flexibility
for modification in the case where $\l$ is small. More precisely, defining the payoff to be $p :\equiv \frac{1}{4(1-\l)} > \frac{1}{4}$ on
$B$ and $1-p <\frac{3}{4}$
on $A$ for a bipartite graph $(G=(A \cup B,E)$, well-spread with parameter $\l$, would work as well and thus solve the problem.
Indeed, the unique independent
set $L$ that maximizes $p(L)$ is $L=B$ in this case, which gives $p(L)=p(B)=|B|/(4(1-\l))= \frac{1}{4}n$. 

There is one thing that needs to be taken care of. Namely, 
in Proposition \ref{prop_bipartite} we assumed that $G=(A \cup B, E)$ has no 
isolated vertices, an assumption that can be made without loss of generality if we only have $2$-element winning coalitions. Now we may have isolated 
vertices that are part of winning coalitions of size~$4$ or larger. But this does not cause any problems either. We simply assign $p:=
\frac{1}{4}$ to these isolated vertices to ensure that indeed all winning coalitions $W$ have $p(W) \ge 1$. Formally, this can also be seen as an extension
of our decomposition: if $G=(A \cup B, E)$ contains isolated vertices, then they are all contained in $B$ (once we assume that $A$
can be completely matched into $B$). So the set of isolated vertices can be seen as a {\lq\lq}degenerate{\rq\rq} well-spread final subgraph 
$(A_k \cup B_k, E_k)$ with $A_k= \emptyset$ and parameter $\l_k=0$. Our proposed payoff $p \equiv \frac{1}{4(1-\l_k)}$ would then
indeed assign $p= \frac{1}{4}$ to all isolated vertices.

It remains to observe that
when we pass to general graphs, no further problems arise. Indeed, all that happens is that vertices in even and odd components 
get payoffs $p=\frac{1}{2}$ which certainly does no harm to the feasibility of $p$. Thus we have proved the following result.

\begin{corollary}\label{cor_not3}
Let $(N,v)$ be a simple game with no minimal winning coalition of size~$3$. Then $\alpha(N,v) \le \frac{1}{4}n$. ~ 
\end{corollary} 

We end this section with the complementary case where all minimal winning coalitions have size~$3$.

\begin{proposition}\label{prop_all3}
 Let $(N,v)$ be a simple game with all minimal winning coalitions of size~$3$. Then $\alpha(N,v) \le \frac{1}{4}n$.
\end{proposition}

\begin{proof}
 We try $p :\equiv \frac{1}{3}$, which is certainly feasible. If this yields $\max p(L) \le \frac{1}{4}n$, we are done.
 Otherwise, there exists a losing coalition $L \subseteq N$ with $p(L) = \frac{1}{3}|L| > \frac{1}{4}n$, or equivalently, $|L|>\frac{3}{4}n$.
 In this case we use an alternative payoff $\tilde{p}$ given by $\tilde{p}\equiv 1$ on $N\backslash L$ and $\tilde{p} \equiv 0$ on $L$.
 Since $|N\setminus L| < \frac{1}{4}n$, this ensures $\tilde{p}(\tilde{L}) <\frac{1}{4}n$ for any losing coalition $\tilde{L}$.
 On the other hand, $\tilde{p}$ is feasible,
 since a winning coalition $W$ cannot be completely contained in $L$, that is, there exists a player $i \in W$ with $\tilde{p}_i=1$
 and hence $\tilde{p}(W) \ge 1$. \qed
\end{proof}

We note that Proposition \ref{prop_all3} is a pure existence result. To compute $\tilde{p}$ it requires to solve a
maximum independent set problem in $3$-uniform hypergraphs, 
which is \NP-hard. 
This can be seen from a reduction from the 
maximum independent set problem in graphs, which is well known to be \NP-hard (see~\cite{GJ79}). Given a graph $G=(V,E)$, construct a
$3$-uniform hypergraph $\bar{G}$ as follows. Add $n=|V|$ new vertices labeled $1, \dots, n$ and extend each edge $e=ij \in E$
to $n$ edges $\{i,j,1\}, \dots, \{i,j,n\}$ in $\bar{G}$. It is readily seen that a maximum independent
set of vertices in $\bar{G}$ (that is, a set of vertices that does not contain any hyperedge) consists of the $n$ new vertices
plus a maximum independent set in~$G$.

\section{Minimal Winning Coalitions of Arbitrary Size}\label{sec_general}

In this section we try to combine the ideas for the two complementary cases to derive an upper bound $\alpha \le \frac{2}{7}$
for the general case. The payoffs $p$ that we consider will all satisfy $p \ge \frac{1}{4}$ so that only winning coalitions of 
size~$2$ and~$3$ are of interest. The basic idea is to start with a bipartite graph $(A\cup B, E)$ representing the 
size~$2$ winning coalitions and a payoff satisfying all these. Standard payoffs that we use satisfy $p \ge \frac{1}{4}$ on $B$ 
and $p \ge \frac{1}{2}$ on $A$. Hence we have to worry only about $3$-element winning coalitions contained in $B$. We seek to
satisfy these by raising the payoff of some vertices in $B$ without spending too much in total.

More precisely, consider a bipartite graph $G=(A \cup B, E)$ representing the winning coalitions of size~$2$. As before, we  
assume that $A$ can be completely matched into $B$, so that our decomposition into well-spread subgraphs $G_i=(A_i \cup B_i,
E_i)$ applies (with possibly the last subgraph $G_k=(A_k \cup B_k, E_k)$ having $A_k=\emptyset$ and $B_k$ consisting of isolated
points as explained at the end of the previous section). Recall the payoff $\bar{\l}_i :\equiv \frac{1}{4(1-\l_i)}$
on $B_i$ and $1-\bar{\l}_i$ on $A_i$ defined for the proof of Corollary \ref{cor_not3}. We first consider the following payoff 
$\bar{p} :\equiv 1-\bar{\l}_i$ on $A_i$ and $\bar{p} :\equiv \bar{\l}_i$ on $B_i$ for $\l_i \ge \frac{1}{4}$, so
$\bar{\l}_i \ge \frac{1}{3}$. For subgraphs with $\l_i < \frac{1}{4}$ (including possibly a final $\l_k=0$) we define 
$\bar{p} \equiv \frac{2}{3}$ on $A_i$ and $\bar{p} \equiv \frac{1}{3}$ on $B_i$. Thus $\bar{p} \ge \frac{1}{3}$ everywhere, 
in particular, $\bar{p}$ is feasible with respect to
all winning coalitions of size at least~$3$. 

Let $\bar{L}$ be a losing coalition with maximum $\bar{p}(L)$.  We define an alternative payoff $\tilde{p}$ as follows: 
For $\l_i \ge \frac{1}{4}$ we set $\tilde{p} :\equiv 1- \bar{\l}_i$ on 
$A_i$, ~$\tilde{p} :\equiv \bar{\l}_i$ on $B\cap \bar{L}$ and $\tilde{p}:\equiv \frac{1}{2}$ on $B_i\backslash \bar{L}$.
For $\l_i <\frac{1}{4}$ we set $\tilde{p} :\equiv \frac{3}{4}$ on $A_i$, ~$\tilde{p} :\equiv \frac{1}{4} $ on $B_i \cap \bar{L}$
and $\tilde{p} :\equiv \frac{1}{2} $ on $B_i \backslash \bar{L}$.

Clearly, both $\bar{p}$ and $\tilde{p}$ are feasible. We claim that a suitable combination of these two yields the desired upper bound.

\begin{lemma}\label{l-propo}
 For $p:=\frac{3}{7}\bar{p}+\frac{4}{7}\tilde{p}$ we get $\alpha = \max_L ~p(L) \le \frac{2}{7}n$.
\end{lemma}

\begin{proof}
 Let $\bar{L}$ as above be a losing coalition with maximum $\bar{p}$-value. Let $\rho_i \in [0,1]$ such that $|\bar{L}\cap B_i|= 
 (1-\rho_i)|B_i| = (1-\rho_i)(1-\l_i) n_i$. For $\l_i \ge \frac{1}{4}$ we then get (using well-spreadedness)
 \begin{equation}\label{eq_barsmalli}
 \bar{p}(\bar{L}_i) \le \left[\rho_i\l_i(1-\bar{\l}_i) +(1-\rho_i)(1-\l_i)\bar{\l}_i\right]n_i. 
 \end{equation}
For $\l_i \ge \frac{1}{4}$, the alternative payoff $\tilde{p}$  equals $\bar{p}$ on $A_i \cup B_i$
except that vertices in $B_i \backslash \bar{L} $ are raised to $\frac{1}{2}$. So a losing coalition $L$ with $L_i := L \cap (A_i\cup B_i)$ 
obviously has $\tilde{p}(L_i) \le \tilde{p}(B_i)$ (as vertices in $B_i$ are relatively more profitable than vertices in $A_i$), \emph{i.e.,}
\begin{equation}\label{eq_tildesmalli}
 \tilde{p}(L_i) = \left[(1-\rho_i)(1-\l_i)\bar{\l}_i+\rho_i(1-\l_i)\frac{1}{2}\right]n_i,
\end{equation}
because, by definition of $\tilde{p}$, exactly $\rho_i(1-\l_i)n_i$ vertices in $B_i$ are raised to $\frac{1}{2}$.  Hence 

\begin{align}\label{eq_smalli}
\frac{3}{7}\bar{p}(\bar{L}_i)+\frac{4}{7}\tilde{p}(L_i) 
&\le \rho_i\left[\frac{3}{7}\l_i(1-\bar{\l}_i)+\frac{4}{7}(1-\l_i)\frac{1}{2}\right]n_i+ 
(1-\rho_i)(1-\l_i)\bar{\l}_i\left(\frac{3}{7}+\frac{4}{7}\right)n_i\\ \nonumber
&\le \rho_i\left[\frac{2}{7}\l_i +\frac{2}{7}(1-\l_i)\right]n_i+(1-\rho_i)\frac{1}{4}n_i \le \frac{2}{7}n_i. \nonumber
\end{align}
where we have used $1-\bar{\l}_i \le \frac{2}{3}$ and $(1-\l_i)\bar{\l}_i=\frac{1}{4}$.\\

For $\l_i < \frac{1}{4}$ (\emph{i.e.,} $\bar{\l}_i < \frac{1}{3})$, we conclude similarly that
\begin{equation}\label{eq_barlargei}
 \bar{p}(\bar{L}_i) \le \left[\rho_i\l_i\frac{2}{3}+(1-\rho_i)(1-\l_i)\frac{1}{3}\right]n_i
\end{equation}
and 
\begin{equation}\label{eq_tildelargei}
 \tilde{p}(L_i) = \left[(1-\rho_i)(1-\l_i)\frac{1}{4} + \rho_i(1-\l_i)\frac{1}{2}\right]n_i.
\end{equation}
Thus,
\begin{align}
 \frac{3}{7}\bar{p}(\bar{L}_i)+\frac{4}{7}\tilde{p}(L_i) 
&\le  \left[\rho_i\left(\frac{2}{7}\l_i+\frac{2}{7}(1-\l_i)\right)+(1-\rho_i)(1-\l_i)\left(\frac{1}{7}+\frac{1}{7}\right)\right]n_i\\
 &\le \left[\rho_i\frac{2}{7}+(1-\rho_i)\frac{2}{7}\right]n_i = \frac{2}{7}n_i.
\end{align}
Now the claim follows by observing that 
$p(L) = \frac{3}{7}\bar{p}(L) +\frac{4}{7}\tilde{p}(L) \le \frac{3}{7}\bar{p}(\bar{L})+\frac{4}{7}\tilde{p}(L)$. \qed 
\end{proof}

Hence we obtained the following theorem.

\begin{theorem}
For every simple game $(N, v)$, $\alpha(N,v)\leq \frac{2}{7}n$.
\end{theorem}

\section{Complete Simple Games}\label{sec_complete}

Recall that a  simple game $(N,v)$ is complete if for a suitable ordering,  say, $1 \succeq 2 \succeq \dots \succeq n$
indicating that $i$ is more powerful than $i+1$ in the sense 
that $v(S\cup \{i\}) \ge v(S \cup \{i+1\})$ for any coalition $S \subseteq N\backslash \{i,j\}$.
Intuitively, the class 
of complete simple games is {\lq\lq}closer{\rq\rq} to weighted voting games than general simple games. The 
next result quantifies this expectation.

\begin{theorem}
  \label{t-csg}
  A complete simple game $(N,v)$ has $\alpha\le  \sqrt{n}\ln n$. 
\end{theorem}

\begin{proof}
  Let $N=\{1,\dots,n\}$ be the set of players and assume without loss of generality that $1\succeq 2\succeq\dots \succeq n$.
  Let $k \in N$ be the largest number such that $\{k, \dots, n\}$ is winning.
  For $i=1, \dots, k$, let $s_i$ denote the smallest size of a winning coalition in $\{i, \dots, n\}$.
  Define $p_i:= 1/s_i$ for $i=1, \dots, k$ and $p_i:=p_k$ for $i=k+1,\dots,n$. Thus, obviously, $p_1 \ge \dots \ge p_k = \dots = p_n$. 
  
  Consider a winning coalition $W\subseteq N$ and let $i$ be the first player in $W$ (with respect to $\succeq$).
  If $|W| \le \sqrt{n}$, then $s_i \le |W| \le \sqrt{n}$ and hence $p(W) \ge p_i = \frac{1}{s_i} \ge \frac{1}{\sqrt{n}}$.
  On the other hand, if $|W| > \sqrt{n}$, then $p(W) > \sqrt{n}p_k \ge \sqrt{n}\frac{1}{n} =\frac{1}{\sqrt{n}}$.
  
  For a losing coalition $L \subseteq N$, we conclude that $|L \cap \{1, \dots , i\}| \le s_i-1$ (otherwise $L$ would dominate
  the winning coalition of size $s_i$ in $\{i, \dots ,n\}$). So $p(L)$ is bounded by 
  $$ \max \sum_{i=1}^{k} x_i\frac{1}{s_i} \text{~~subject to~~} \sum_{j=1}^{i}x_j \le s_i-1,~ i=1, \dots , k.$$
  The optimal solution of this maximization problem is easily seen to be $x_1=s_1-1, x_i=s_i-s_{i-1} \text{~for~} 2 \le i \le k$.
  Hence 
  \begin{align} \nonumber
  p(L) &\le (s_1-1)\frac{1}{s_1} + (s_2-s_1)\frac{1}{s_2} + \dots + (s_k-s_{k-1})\frac{1}{s_k}\\ \nonumber
       &\le \frac{1}{2} +   \dots + \frac{1}{s_k} \le \ln n. \nonumber
  \end{align}
  
  Summarizing, we obtain $p(L)/p(W) \le \sqrt{n} \ln n$, as claimed. \qed
  \end{proof}

In \cite{paper_alpha_roughly_weighted} it is conjectured that $\alpha = O(\sqrt{n})$ holds for complete simple games.
We direct the reader to \cite{paper_alpha_roughly_weighted} for further details, including a lower bound of order $\sqrt{n}$
as well as specific subclasses of complete simple games for which $\alpha=O(\sqrt{n})$ can be proven.

\section{Algorithmic Aspects}\label{sec_algo}

A fundamental question concerns the complexity of our original problem (\ref{eq_minmax}). For general simple games this depends on how the  game in question is given, and we refer to Section~\ref{sec_introduction} for a discussion.
Here we concentrate on the {\lq\lq}graphic{\rq\rq} case where the minimal winning coalitions are given as the edges of a graph $G$.

\begin{proposition}
 For a bipartite graph $G=(N,E)$ we can compute $\alpha_G$ in polynomial time.
\end{proposition}

\begin{proof}
 Let $P \subseteq \R^n$ denote the set of feasible payoffs (satisfying $p \ge 0$ and $p_i+p_j \ge 1$ for $ij \in E$). For $\alpha \in \R$ we let
 $$P_\alpha := \{p\in P ~|~ p(L) \le \alpha \text{~for all independent ~} L \subseteq N\}.$$
Thus $\alpha_G = \min \{\alpha ~|~P_\alpha \neq \emptyset\}$. The separation problem for $P_\alpha$ (for any given $\alpha$) is efficiently solvable.
Given $p\in \R^n$, we can check feasibility and we can check whether $\max \{p(L) ~| ~L \subseteq N \text{~independent} \} \le \alpha $ by solving a 
corresponding maximum weight independent set problem in the bipartite graph $G$. Thus we can, for any given $\alpha \in \R$, apply the ellipsoid
method to either compute some $p \in P_\alpha$ or conclude that $P_\alpha =\emptyset$. Binary search then exhibits the minimum value for which 
$P_\alpha$ is non-empty. Note that binary search works indeed in polynomial time since the optimal $\alpha$ has size polynomially bounded in $n$. 
The latter follows by observing that 
\begin{equation}\label{eq_LP}
\alpha = \min \{a ~|~ p_i+p_j \ge 1~ ~\forall ij \in E,~ p(L)-a \le 0 ~ ~\forall L \subseteq N \text{~independent}, p \ge 0\}
\end{equation}
can be computed by solving a linear system of $n$ constraints defining an optimal basic solution of the above linear program. \qed \end{proof}

The above proof also applies to all other classes of graphs, such as claw-free graphs and generalizations thereof (see~\cite{brand2016claw})
in which finding a  weighted maximum independent set is polynomial-time solvable. In general, however, computing $\alpha$ is \NP-hard (just like computing a maximum independent set).

\begin{proposition}
 Computing $\alpha_G$ for arbitrary graphs $G$ is \NP-hard. 
\end{proposition}

\begin{proof}
 Given $G=(N,E)$ with maximum independent set of size $k$, let $G'=(N',E')$ and $G''=(N'',E'')$ be two disjoint copies of $G$. For each $i'\in N'$
 and $j''\in N''$ we add  an edge $i'j''$ if and only if $i=j$ or $ij \in E$  and call the resulting graph $G^*= (N^*, E^*)$. (In graph theoretic terminology
 $G^*$ is also known as the \emph{strong product} of $G$ with $P_2$.) We claim that $\alpha_{G^*} = k/2$ (thus showing that computing $\alpha_{G^*}$
 is as difficult as computing $k$).
 
 First note that the independent sets  in $G^*$ are exactly  the sets $L^* \subseteq N^*$ that arise from an independent set $L \subseteq N$ in $G$
 by splitting $L$ into two complementary sets $L_1$ and $L_2$ and defining $L^*:= L_1'\cup L_2''$.
 Hence, $p\equiv \frac{1}{2}$ on $N^*$ yields $\max p(L^*) = k/2$ where the maximum is taken over all independent sets $L^* \subseteq N^*$ in $G^*$. 
 This shows that $\alpha_{G^*} \le k/2$.
 
 Conversely, let $p^*$ be any feasible payoff in $G^*$ (that is, $p^* \ge 0$ and $p^*_i+p^*_j \ge 1 $ for all $ij \in E^*$). Let $L \subseteq N$ 
 be a maximum independent set of size $k$ in $G$ and construct $L^*$ by 
 including for each $i \in L$ either $i'$ or $i''$ in $L^*$, whichever has $p$-value at least $\frac{1}{2}$. Then, by construction, $L^*$ is an 
 independent set in $G^*$ with $p^*(L^*) \ge k/2$, showing that $\alpha_{G^*} \ge k/2$. \qed
\end{proof}

Summarizing, for graphic simple games, computing $\alpha_G$ is
as least as hard as computing the size of a maximum independent in~$G$.
For our last result we assume that $a$ is a fixed integer, that is, $a$ is not part of the input.

\begin{proposition}
 For every fixed $a > 0$, it is possible to decide if $\alpha_G \leq a$ in polynomial time for an arbitrary graph $G=(N,E)$.
\end{proposition}

\begin{proof} 
Let $k= 2\lceil a+\epsilon \rceil$ for some $\epsilon >0$. By brute-force, we can check in $O(n^{2k})$ time if $N$ contains 
  $2k$ vertices $\{u_1,\ldots,u_k\}\cup \{v_1,\ldots,v_k\}$ that induce $k$ disjoint copies of $P_2$, that is, paths $P_i=u_iv_i$ of length $2$ for 
  $i=1,\ldots,k$ with no edges joining any two of these paths.
If so, then the condition $p(u_i)+p(v_i)\geq 1$ implies that one of $u_i,v_i$, say $u_i$, must receive a payoff $p(u_i)\geq \frac{1}{2}$,
and hence $U=\{u_1,\ldots,u_k\}$ has $p(U)\geq k/2 > a$. As $U$ is an independent set, we conclude that $\alpha(G)>a$.

Now assume that $G$ does not contain $k$ disjoint copies of $P_2$ as an induced subgraph, that is,
 $G$ is $kP_2$-free. For every $s\geq 1$, the number of maximal independent sets in a $sP_2$-free graphs is $n^{O(s)}$ 
due to a result of Balas and Yu~\cite{BY89}. Tsukiyama, Ide, Ariyoshi, and Shirakawa~\cite{TIAS77} show how to enumerate 
all maximal independent sets of a graph $G$ on $n$ vertices and $m$ edges using time $O(nm)$ per independent set. Hence we can find all maximal 
independent 
sets of $G$ and thus solve, in polynomial time, the linear program \ref{eq_LP}. Then it remains to check if the 
solution found satisfies $\alpha\leq a$. \qed
\end{proof}
  
\section{Conclusions}\label{sec_conclusion}

The two main open problems are  to prove the upper bound of $\frac{1}{4}n$ for all simple games and to tighten
the upper bound for complete simple games to $O(\sqrt{n})$. 
In order to classify simple games, many more subclasses of simple games have been identified in the literature. 
Besides the two open problems, no optimal bounds for $\alpha$
are known for other subclasses of simple games, such as \textit{strong}, 
\textit{proper}, or \textit{constant-sum} games, that is, where $v(S)+v(N\backslash S)\ge 1$, $v(S)+v(N\backslash S)\le 1$, 
or $v(S)+v(N\backslash S)= 1$ for all $S\subseteq N$, respectively.

\medskip
\noindent
{\it Acknowledgments.} The second and fourth author thank P\'eter Bir\'o and Hajo Broersma for fruitful discussions on the topic of the paper.

 \bibliographystyle{abbrv}

\appendix

\section{Finding a Decomposition into Well-Spread Graphs}\label{a-proof}

As mentioned, for the efficient implementation of the procedure for splitting a bipartite graph into well-spread subgraphs,
all we need to solve is $\max_{S \subseteq A} |S|/|N(S)|$ in bipartite graphs $G=(A\cup B, E)$, and this is equivalent to minimizing the submodular function $f(S)=r|N(S)|-|S|$. Instead of using a known algorithm for solving the latter, we present a direct algorithm.

\begin{lemma}
 Consider a bipartite graph 
 $G=(A\cup B,E)$ of order $n$ such that $A$ can be matched into $B$. Then we can find $\max_{S \subseteq A} |S|/|N(S)|$ 
 in time $O(n^6 \log n)$.
\end{lemma}

\begin{proof}
 Let $0 < r_1 < r_2 < \dots < r_k \le 1$ be a complete list of all fractions in $[0,1]$ of the form $r=p/q$ with $p,q \in \{1, \dots, n\}$.
 We compute $\max_{S \subseteq A} |S|/|N(S)|$ by binary search. To check whether there exists  $S \subseteq A$ with $|S|/|N(S)|>p/q$, 
 we construct a bipartite graph $\bar{G}=(\bar{A}\cup\bar{B}, \bar{E})$, where $\bar{A}$ consists of $q$ disjoint copies of $A$, $\bar{B}$ 
 consists of $p$ disjoint
 copies of $B$, and each copy of $A$ is connected to each copy of $B$ in exactly the same way as $A$ is connected to $B$ in $G$.
 
 We claim that 
 \begin{equation}\label{eq_S}
  \exists S \subseteq A: |S|/|N(S)| >p/q
 \end{equation}
 is equivalent to 
 \begin{equation}\label{eq_barS}
  \exists \bar{S}\subseteq\bar{A}: |\bar{S}| > |\bar{N}(\bar{S})|,
 \end{equation}
where $\bar{N}(\bar{S})$ is the neighborhood of $\bar{S} \subseteq \bar{A}$ in $\bar{G}$. 

Indeed, if (\ref{eq_S}) holds, let $\bar{S} \subseteq \bar{A}$ consist of all $q$ copies of $S$, so that $|\bar{S}|=q|S|$. The neighborhood of 
$\bar{S}$ in $\bar{G}$ then consists of all $p$ copies of $N(S)$, so $|\bar{N}(\bar{S})|=p|N(S)|$, thus~(\ref{eq_barS})
holds. Conversely, if (\ref{eq_barS}) holds and $\bar{S}\subseteq \bar{A}$ satisfies $|\bar{S}| > |\bar{N}(\bar{S})|$, we may assume without loss of generality
that $\bar{S}$ consists of $q$ copies of some set $ S \subseteq A$. (Indeed, note that if $\bar{S}$ contains any copy of some $i \in A$, we may add
all other copies of $i$ to $\bar{S}$ without affecting $\bar{N}(\bar{S})$ - and hence without affecting $|\bar{S}| > |\bar{N}(\bar{S})|$.) But then 
$\bar{N}(\bar{S})$ simply consists of all $p$ copies of $N(S)$ and we get (\ref{eq_S}).

Since (\ref{eq_barS}) can be decided by solving a matching problem in $\bar{G}$, a graph of size $n^2$, this finishes the proof. (Recall that
matching problems of size $n$ can be solved in time $O(n^3)$ (see, for example~\cite{lovasz2009matching}). \qed
\end{proof}

  \end{document}